\documentclass[A4paper,USenglish]{lipics-v2016}
\usepackage{microtype}
\usepackage{graphicx}
\usepackage{color}
\newcommand{\hide}[1]{}
\newcommand{\QED}{\hfill$\qed$}
\newtheorem{algorithm}{Algorithm}

\newtheorem{proposition}{Proposition}

\title{Improved Algorithms for Computing $k$-Sink on Dynamic Path Networks}%
\titlerunning{$k$-sink on dynamic path networks}

\author[1]{Binay Bhattacharya}
\author[2]{Mordecai J. Golin}
\author[3]{Yuya Higashikawa}
\author[4]{Tsunehiko Kameda}
\author[5]{Naoki Katoh}
\affil[1,4]{School of Computing Science, Simon Fraser University,
  Vancouver, Canada\\
  \texttt{binay@sfu.ca, tikokameda@gmail.com}}
  \affil[2]{Dept. of Computer Science, Hong Kong Univ. of Science and Technology, Hong Kong\\
  \texttt{golin@cse.ust.hk}}
  \affil[3]{Dept. of Information and System Engineering, Chuo University, Tokyo, Japan \\
  \texttt{higashikawa.874@g.chuo-u.ac.jp}}
  \affil[5]{School of Science and Technology, Kwansei Gakuin University, Hyogo, Japan \\
  \texttt{naoki.katoh@gmail.com}}
\authorrunning{B.~Bhattacharya~et~al.}
\Copyright{Binay Bhattacharya, Mordecai J. Golin, Yuya Higashikawa, Tsunehiko Kameda and Naoki Katoh}
\keywords{Facility location, $k$-sink, parametric search, dynamic path network}

\begin{document}

\maketitle

\begin{abstract}
We present a novel approach to finding the $k$-sink
on dynamic path networks with general edge capacities.
Our first algorithm runs in $O(n\log n + k^2\log^4 n)$ time, 
where $n$ is the number of vertices on the given path,
and our second algorithm runs in $O(n\log^3 n)$ time.
Together, they improve upon the previously most efficient $O(kn\log^2n)$ time algorithm due to
Arumugam et al.~\cite{arumugam2016}
for all values of $k$.
In the case where all the edges have the same capacity,
we again present two algorithms that run in $O(n + k^2\log^2n)$ time
and $O(n\log n)$ time, respectively,
and they together improve upon the previously best $O(kn)$ time
algorithm due to Higashikawa et al.~\cite{higashikawa2015a}
for all values of $k$.
\end{abstract}


\section{Introduction}\label{sec:intro}
Investigation of evacuation problems dates back many years~\cite{hamacher2002,mamada2002}.
The $k$-sink problem is to locate $k$ sinks in such a way that every evacuee can
evacuate to a sink as quickly as possible,
when disasters, such as earthquakes and tsunamis, strike.
The problem can be modeled by a network whose vertices represent the places where
the evacuees are initially located and the edges represent possible evacuation routes.
Associated with each edge is the transit time across it in either direction and its capacity
in terms of the number of people who can enter it per unit time~\cite{hamacher2002}.
Madama et al.~\cite{mamada2006} solved this problem for the dynamic tree networks in $O(n\log^2 n)$ time
under the condition that only a vertex can be a sink.
For the 1-sink problem in the dynamic tree networks with uniform edge capacities,
Higashikawa et al. proposed an $O(n\log n)$ algorithm~\cite{higashikawa2014f}
with the condition that the sink can be either at a vertex or on an edge.

On dynamic path networks with uniform edge capacities,
it is straightforward to compute the 1-sink  in linear time~\cite{cheng2013}.
The $k$-sink problem for dynamic path networks with general 
and uniform edge capacities was solved in $O(kn \log^2n)$ time by Arumugam et al.~\cite{arumugam2016}
and in $O(kn)$ time by Higashikawa et al.~\cite{higashikawa2015a}, respectively.

In this paper we present two algorithms for the $k$-sink problem for the dynamic path networks
with general edge capacities.
Together, they outperform all other known algorithms.
We also present two algorithms for the dynamic path networks
with uniform edge capacities.
All our algorithms consists of two levels: feasibility tests at the lower level,
and optimization at the higher level, making use of feasibility tests.
Our results presented in this paper are the first algorithms that run in sub-quadratic time
in $n$,
regardless of the value of $k$, which can grow with $n$.

This paper is organized as follows.
In the next section, we define our model and the terms that are used throughout
the paper.
In Sec.~\ref{sec:overview}, we give an overview of our algorithm.
Sec.~\ref{sec:capacitated} introduces a data structure called the critical cluster tree,
which plays a central role in the rest of the paper.
In Sec.~\ref{sec:2tasks}, we identify two important tasks that form building blocks of
our algorithms.
and also discuss feasibility test.
Sec.~\ref{sec:optimization} presents several algorithms for uniform and general edge capacities.
Finally, Sec.~\ref{sec:conclusion} concludes the paper.


\section{Preliminaries}\label{sec:prelim}
Let $P = (V,E)$ be a path network,
whose vertices $v_1,v_2,\ldots, v_n$ are arranged from left to right
in this order.
For $i=1,2,\ldots, n$,
vertex $v_i$ has an integral weight $w_i (>0)$,
representing the number of evacuees,
and each edge $e_i=(v_i,v_{i+1})$ has a fixed non-negative length (distance)
$l_i$ and {\em capacity} $c_i$.
We assume that all evacuees from a vertex evacuate to the same sink.
We also assume that a sink has infinite capacity,
so that the evacuees coming from the left and right of a sink do not interfere with each other.
By $x\in P$, we mean that point $x$ lies on either an edge or a vertex of $P$.
For a vertex $v$, $v^+$ (resp. $v^-$) denotes the point just to the right (resp. left)
of vertex $v$ that is arbitrarily close to $v$.
For $a,b\in P$, $a\prec b$ or $b\succ a$ means that $a$ lies to the left of $b$.
Let $d(a,b)$ denote the distance between $a$ and $b$.
If $a$ and/or $b$ lies on an edge, we use the prorated distance.
The transit time for a unit distance is denoted by $\tau$,
so that it takes $d(a,b)\tau$ time to travel from $a$ to $b$.
Let $c(a,b)$ denote the minimum capacity of the edges on the subpath
of $P$ between $a$ and $b$.
Let $V[a,b]$ denote the set of vertices on the path from $a$ to $b$.
The subpath from $a$ to $b$, including $a$ and $b$,
is denoted by $P[a,b]$.
If $a$ (resp. $b$) is excluded, we use $P(a,b]$ (resp. $P[a,b)$).
Define
\begin{eqnarray}
W[v_i,v_j] &=& \sum_{v_l\in V[v_i,v_j]} w_l. \label{eqn:weightarray}
\end{eqnarray}
Clearly $W[v_i,v_j]$ can be computed in constant time once we construct the array
$\{W[v_1,v_j]\mid j=1,2,\ldots, n\}$ in $O(n)$ time.

Given a subpath $P[v_i,v_j]$ of a dynamic path network $P$ and a sink $s \in P[v_i,v_j]$,
let $\Theta(s,[v_i,v_j])$ denote the evacuation time to $s$ for the evacuees on $P[v_i,v_j]$.
We also define the {\em L-cost} (resp. {\em R-cost}) of vertex $v_h\in V[v_i,v_j]$
in $P[v_i,v_j]$,
as seen from $s\succ v_j$ (resp. $s\prec v_i$)
to be the least evacuation time to $s$ for all the evacuees on the vertices on 
$P[v_i,v_h]$ (resp. $P[v_h,v_j]$),
assuming that they all arrive at $s$ as continuously as possible.
For any vertex $v_h \in V[v_i,v_j]$, its L-cost and R-cost are thus
\begin{eqnarray}
\theta_L(s,[v_i,v_h]) &=& d(v_h,s) \tau + \frac{W[v_i,v_h]}{c(v_h,s)} \ \mbox{~for~} \ x\succ v_j, \label{eqn:leftcost3}\\
\theta_R(s,[v_h,v_j]) &=& d(s, v_h)\tau + \frac{W[v_h,v_j]}{c(s,v_h)} \ \mbox{~for~} \ x\prec v_i. \label{eqn:rightcost3}
\end{eqnarray}
Note that each of these functions is linear in the distance to $s$.
\begin{lemma}\label{lem:Thetas-capacitated}{\rm \cite{arumugam2016,higashikawa2014e}}
Given a subpath $P[v_i,v_j]$ of a dynamic path network $P$ and a sink $s \in P[v_i,v_j]$,
$\Theta(s,[v_i,v_j])$ is represented by the following formula:
\begin{eqnarray}\label{eqn:completion-timeB}
\Theta(s,[v_i,v_j]) = \max \left\{ \max_{v_h \in V[v_i,s]} \theta_L(s, [v_i, v_h]), \max_{v_h \in V[s,v_j]} \theta_R(s, [v_h, v_j]) \right\}. \label{eqn:cost1}
\end{eqnarray}
\end{lemma}
A problem instance is said to be {\em $(t,k)$-feasible}
if exist $k$ sinks such that every evacuee can reach a sink within time $t$.
In our algorithms proposed in this paper,
we perform preprocessing to construct a useful data structure,
which makes $(t,k)$-feasibility test efficient.


\section{Overall strategy of our algorithm}
\label{sec:overview}
To carry out $(t,k)$-feasibility test, we repeatedly solve the following problems:

\medskip
\noindent
$L$-$test(P[v_i,v_j],t)$: It returns {\bf yes} if  every evacuee on a subpath $P[v_i,v_j]$ can reach a sink $v_j^+$ within time $t$. Otherwise it returns {\bf no}.
\medskip

\noindent
$R$-$test(P[v_i,v_j],s,t)$: It returns {\bf yes} if every evacuee on a subpath $P[v_i,v_j]$ can reach a sink $s$ within time $t$ where $s$ is located on $(v_{i-1}, v_i]$. Otherwise it returns {\bf no}.
\medskip

As will be seen in Sec.~\ref{sec:capacitated}, each of $L$-$test(P[v_i,v_j],t)$ and $R$-$test(P[v_i,v_j],s,t)$ can be done in $O(\log^2 n)$ time after constructing the data structure (called {\it critical cluster tree}). The critical cluster tree is a balanced binary search tree with height $O(\log n)$. 
\medskip

\subsection{Feasibility test}
The $(t,k)$-feasibility can be tested as follows:
We first compute
\begin{equation}
l_1= \max\{ j \mid 1 \le j \le n, \ L\mbox{-}test(P[v_1,v_j],t) \mbox{ is "yes"}\}.
\end{equation}
and find a sink $s_1 \in (v_{l_1}, v_{l_1+1}]$.
We then compute
\begin{equation}
r_1= \max\{ j \mid l_1+1 \le j \le n, \ R\mbox{-}test(P[v_{l_1+1},v_j],s_1,t) \mbox{ is "yes"}\}.
\end{equation}
Repeating this procedure, if we eventually obtain $r_k=n$, $(t,k)$-feasibility test succeeds. If $r_k< n$, it fails.
In fact, as will be seen in Sec.~\ref{sec:feasibility}, $(t,k)$-feasibility test can be done in $O(k\log^3 n)$ time.

Using $(t,k)$-feasibility test as a subroutine, we can find a minimum value $t^*$ such that $(t^*,k)$-feasibility test succeeds, which gives us an optimal evacuation time. This can be done by executing $(t,k)$-feasibility tests in binary search fashion.

\subsection{Machineries in the data structure}
The key idea is to use the balanced binary search tree $\cal T$ with appropriate information stored at each node of the tree 
which enables us to execute each of $L$-$test(P[v_i,v_j],t)$ and $R$-$test(P[v_i,v_j],s,t)$ in $O(\log^2 n)$ time.

For $L$-$test(P[v_i,v_j],t)$, we need to compute 
\begin{equation}\label{eq:query1}
\Theta_L(v_i,v_j)=\max_{v_h \in V[v_i,v_j]} \theta_L(v_j^+, [v_i, v_h]).
\end{equation}
Also for $R$-$test(P[v_i,v_j],s,t)$, we need to compute 
\begin{equation}\label{eq:query1.1}
\Theta_R(v_i,v_j,s)=\max_{v_h \in V[v_i,v_j]} \theta_R(s, [v_h, v_j]).
\end{equation}
$L$-$test(P[v_i,v_j],t)$ succeeds if and only if $\Theta_L(v_i,v_j) \le t$ 
and $R$-$test(P[v_i,v_j],s,t)$ succeeds if and only if $\Theta_R(v_i,v_j,s) \le t$.

For leaf nodes $l(v_i)$ and $l(v_j)$ in $\cal T$ which correspond to $v_i$ and $v_j$, respectively, let $u$ be the least common ancestor of $\cal T$. Then in the subtree ${\cal T}(u)$ with the root $u$ in $\cal T$, we can identify 
a set of vertex-disjoint subpaths which covers vertices of $P[v_i,v_j]$ such that the number of such subpaths is  $O(\log n)$, and every subpath corresponds to the set of leaves that a subtree ${\cal T}(u')$ spans for some node $u'$ in ${\cal T}(u)$.
Let $\mathcal{P}[v_i,v_j]$ denote the set of such subpaths.

In the rest of this section, we only show how to compute $\Theta_L(v_i,v_j)$ since $\Theta_R(v_i,v_j,s)$ is symmetric, so can be similarly computed.
The computation of (\ref{eq:query1}) reduces to 
\begin{equation}\label{eq:query2}
\Theta_L(v_i,v_j)=\max_{P[v_l,v_r] \in \mathcal{P}[v_i,v_j]} \left\{ \max_{v_h \in V[v_l,v_r]} \theta_L(v_j^+, [v_i, v_h]) \right\}.
\end{equation}
To evaluate (\ref{eq:query2}), 
we need to compute 
\begin{equation}\label{eq:query3}
\max_{v_h \in V[v_l,v_r]} \theta_L(v_j^+, [v_i, v_h]) = \max_{v_h \in V[v_l,v_r]} \left\{d(v_h,v_j^+)\tau + \frac{W[v_i,v_h]}{c(v_h,v_{j+1})}\right\}
\end{equation}
for every subpath $P[v_l,v_r] \in \mathcal{P}[v_i,v_j]$.
Since $v_i \le v_l \le v_r \le v_j$,
the right side of (\ref{eq:query3}) is rewritten as 
\begin{equation}\label{eq:query4}
\max_{v_h \in V[v_l,v_r]} \left\{d(v_h,v_r)\tau + d(v_r,v_j^+)\tau + \frac{W[v_i,v_{l-1}] +W[v_l, v_h]}{\min\{c(v_h,v_r), c(v_r, v_{j+1})\}}\right\}.
\end{equation}
Suppose that $u'$ is a node of $\cal T$ spanning $P[v_l,v_r]$.
Then, to facilitate the computation of (\ref{eq:query4}) for general case, we will prepare at node $u'$ the following machinery that allows us to compute
\begin{equation}\label{eq:query5}
cost^L_{u'}(W,C)=\max_{v_h \in V[v_l,v_r]} \left\{d(v_h,v_r)\tau + \frac{W +W[v_l, v_h]}{\min\{c(v_h,v_r), C\}}\right\}
\end{equation}
in $O(\log n)$ time once $W$ and $C$ are given. Here $W$ and $C$ are unknown parameters. 
This part will be explained in more detail in Sec.~\ref{sec:capacitated}.


\section{Data structures for the edge-capacitated case}\label{sec:capacitated}
We want to perform $(t,k)$-feasibility tests for many different values of completion time $t$.
Therefore, it will be useful to spend some time during preprocessing to
construct data structures which facilitate those tests.
Let us consider an arbitrary subpath $P[v_i,v_j]$,
where $i\leq j$.
The vertex $v_k \in V[v_i,v_j]$ that maximizes $\theta_L(x,[v_i,v_k])$ at $x\succ v_j$
(resp. $\theta_R(v_{i-1},[v_k,v_j])$ at $x\prec v_i$)
is called the {\em L-critical vertex} (resp. {\em R-critical vertex}) of $P[v_i,v_j]$ {w.r.t. $x$},
and the corresponding subpath $P[v_i,v_k]$ (resp. $P[v_k,v_j]$) is called
the {\em L-critical cluster} (resp. {\em R-critical cluster}) of $P[v_i,v_j]$ {w.r.t. $x$}.
It is easy to show the following proposition.
\begin{proposition}
The L-critical (resp. R-critical) vertex/cluster w.r.t. $x$ is the same for all points $x$ on an edge,
excluding its left (resp. right) end vertex.
\QED 
\end{proposition}
Therefore, we can talk about a critical vertex/cluster w.r.t. an edge.
The L-critical vertex of $P[v_i,v_j]$ w.r.t. edge $(v_j,v_{j+1})$
(resp. R-critical vertex of $P[v_i,v_j]$ w.r.t. edge $(v_{i-1},v_i)$) 
is denoted by $c_L^{[i,j]}$ (resp. $c_R^{[i,j]}$),
and L-critical cluster of $P[v_i,v_j]$ w.r.t. edge $(v_j,v_{j+1})$ (resp. R-critical cluster of $P[v_i,v_j]$ w.r.t. edge $(v_{i-1},v_i)$)
is denoted by $C_L^{[i,j]}$ (resp. $C_R^{[i,j]}$).
We thus have $C_L^{[i,j]}=P[v_i,c_L^{[i,j]}]$  and $C_R^{[i,j]}=P[c_R^{[i,j]},v_j]$.

\subsection{Critical cluster tree}\label{sec:datastructure2}
We first construct the {\em critical cluster tree} (or {\em CC-tree} for short), ${\cal T}$,
with root $\rho$, whose leaves are the vertices of $P$,
arranged from left to right.
It is a balanced tree with height $O(\log n)$.
In balancing, the vertex weights are not considered.
See Fig.~\ref{fig:CCtree},
where $\pi(v_i,\rho)$ denotes the path from $v_i$ to root $\rho$.
\begin{figure}[ht]
\centering
\includegraphics[height=3cm]{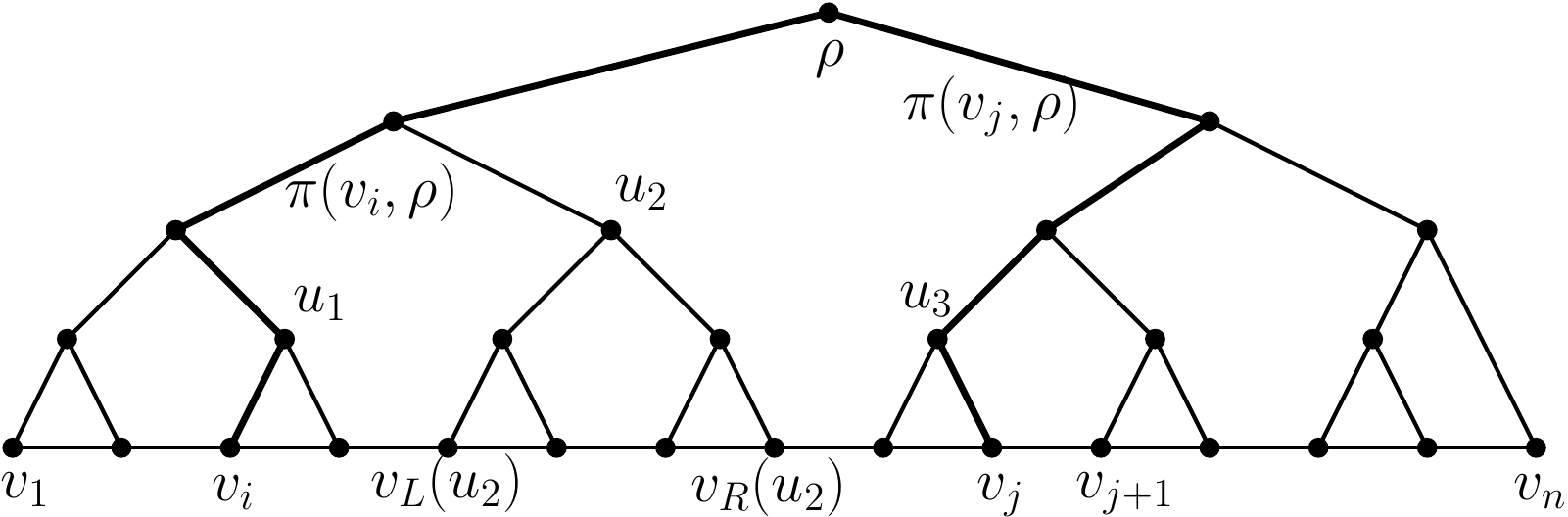}
\caption{
CC-tree ${\cal T}$.}
\label{fig:CCtree}
\end{figure}
For a node\footnote{We use the term
``node'' here to distinguish it from the vertices on the path.
A vertex, being a leaf of ${\cal T}$, is considered a node, but an interior node of ${\cal T}$
is not a vertex.}
 $u$ of ${\cal T}$, let ${\cal T}(u)$ denote the subtree rooted at $u$,
let $u_l$ (resp. $u_r$) be its left (resp. right) child node,
and let $v_L(u)$ (resp. $v_R(u)$) denote the leftmost (resp. rightmost) vertex on $P$
that belongs to ${\cal T}(u)$.
We say that ${\cal T}(u)$ {\em spans} subpath $P[v_L(u),v_R(u)]$ and also $u$ spans $P[v_L(u),v_R(u)]$.
At node $u$, we store two sorted lists of capacities $c(v_h,v_R(u))$ for $v_h \in V[v_L(u),v_R(u)] \setminus \{v_R(u)\}$
and $c(v_L(u),v_h)$ for $v_h \in V[v_L(u),v_R(u)] \setminus \{v_L(u)\}$.
The list of $c(v_h,v_R(u))$ can be computed in the decreasing order in $h$ in $O(|{\cal T}(u)|)$ time.
Symmetrically, the list of $c(v_L(u),v_h)$ can be also computed in $O(|{\cal T}(u)|)$ time.

For $L$-$test(P[v_i,v_j],t)$ and $R$-$test(P[v_i,v_j],s,t)$ mentioned in Sec.~\ref{sec:overview},
we need to determine $c_L^{[i,j]}$ and $c_R^{[i,j]}$, respectively.
To do this, for every highest node $u$ of $\cal T$ spanning a subpath of $P[v_i,v_j]$
we will prepare a machinery at each node $u$ of $\cal T$ that allows us to compute
\begin{eqnarray}
cost^L_u(W,C)=\max_{v_h \in V[v_L(u),v_R(u)]} \left\{d(v_h,v_R(u))\tau + \frac{W +W[v_L(u), v_h]}{\min\{c(v_h,v_R(u)), C\}}\right\}, \label{eq:lquery}\\
cost^R_u(W,C)=\max_{v_h \in V[v_L(u),v_R(u)]} \left\{d(v_L(u),v_h)\tau + \frac{W[v_h,v_R(u)]+W}{\min\{C, c(v_L(u),v_h)\}}\right\}. \label{eq:rquery}
\end{eqnarray}
in $O(\log n)$ for arbitrary $W$ and $C$.

Suppose that we have such a machinery at each node $u$.
Once $P[v_i,v_j]$ is given, for a node $u$ spanning a subpath of $P[v_i,v_j]$,
$W$ and $C$ are given as $W[v_i,v_{l-1}]$ (where $v_l = v_L(u)$) and $c(v_R(u),v_j)$ in (\ref{eq:lquery}), respectively.
Let us call a vertex which achieves the maximum value in (\ref{eq:lquery}) 
w.r.t. $u$, $W=W[v_i,v_{l-1}]$ and $C=c(v_R(u),v_j)$ 
{\em an L-critical candidate} 
of $P[v_i,v_j]$.
Then one of L-critical candidates 
of $P[v_i,v_j]$ must be $c_L^{[i,j]}$,
which implies that we can do an $L$-$test(P[v_i,v_j],t)$ in $O(\log^2 n)$ time
since there are $O(\log n)$ nodes spanning vertex-disjoint subpaths of $P[v_i,v_j]$.
Similarly, {\em an R-critical candidate} of $P[v_i,v_j]$ is defined, and one of R-critical candidates of $P[v_i,v_j]$ is $c_R^{[i,j]}$,
thus $R$-$test(P[v_i,v_j],s,t)$ can be done in $O(\log^2 n)$ time.

To compute $cost^L_u(W,C)$, 
one idea is to prepare a two-dimensional table at node $u$ which returns the L-critical vertex of $P[v_i,v_j]$ for queries of $W$ and $C$,
but it takes much time and space in total.
Instead of this, we actually store at node $u$ two linear tables of vertices in $V[v_L(u),v_R(u)]$:
one returns a vertex $v^1_u(W)$ for a query of $W$ which achieves
\begin{eqnarray}
\max_{v_h \in V[v_L(u),v_R(u)]} \left\{d(v_h,v_R(u))\tau + \frac{W +W[v_L(u), v_h]}{c(v_h,v_R(u))}\right\}, \label{eq:lquery2}
\end{eqnarray}
and the other one returns a vertex $v^2_u(C)$ for a query of $C$ which achieves
\begin{eqnarray}
\max_{v_h \in V[v_L(u),v_R(u)]} \left\{d(v_h,v_R(u))\tau + \frac{W[v_L(u), v_h]}{C}\right\}. \label{eq:lquery3}
\end{eqnarray}
We call the first table {\em the left weight table} of $u$ and the other {\em the left capacity table} of $u$.
Also, to compute $cost^R_u(W,C)$, we store similar two tables at node $u$, 
called {\em the right weight table} of $u$ and {\em the right capacity table} of $u$.
Note that for each leaf node $u=v_i$ (which is a vertex of $P$),
tables always return $v_i$ itself for any $W$ and $C$.
In Sec.~\ref{sec:2algorithms}, we will show how to use these tables to compute $(\ref{eq:lquery})$ and $(\ref{eq:rquery})$.

\subsection{CC-tree construction}\label{sec:CCtree}
In this section, we show how to construct the left weight table and the left capacity table at a node $u$ of $\cal T$.
Note that in the construction of CC-tree, we can construct tables of $u$ without using any information stored at children $u_l$ and $u_r$.

\medskip\noindent
(a) {\bf Weight table:} For a vertex $v_h \in V[v_L(u),v_R(u)]$, let $f^1_h(W)$ denote a function of $W$ such that
\begin{eqnarray}
f^1_h(W) = \alpha^1_hW + \beta^1_h, \label{eq:lquery2.1}
\end{eqnarray}
where $\alpha^1_h =1/c(v_h,v_R(u))$ and $\beta^1_h=d(v_h,v_R(u))\tau + W[v_L(u), v_h]/c(v_h,v_R(u))$.
Then the equation (\ref{eq:lquery2}) can be rewritten as
\begin{eqnarray}
\max_{v_h \in V[v_L(u),v_R(u)]} f^1_h(W). \label{eq:lquery2.2}
\end{eqnarray}
Note that once we compute the upper envelope of $f^1_h(W)$ for all $v_h \in V[v_L(u),v_R(u)]$,
it can return a vertex $v^1_u(W)$ for a query of $W$ which achieves (\ref{eq:lquery2.2}),
which is equivalent to the left weight table of $u$.
Here $f^1_h(W)$ is a linear function in $W$ and $\alpha^1_h$ is decreasing in $h$.
Using the concept of duality of lines and points in 2-D,
it is known that computing the upper envelope of lines is equivalent to
computing the lower convex hull
of points \cite{deberg2008,preparata2012}.
As noted in \cite{preparata2012}, it is known that if points are sorted in
$x$-coordinates, the convex hull can be computed in linear time by
using the Graham scan algorithm \cite{graham1972}.
Summarizing these facts, we can obtain the left weight table of $u$ in $O(|{\cal T}(u)|)$ time,

\medskip\noindent
(b) {\bf Capacity table:} For a vertex $v_h \in V[v_L(u),v_R(u)]$, let $f^2_h(1/C)$ denote a function of $1/C$ such that
\begin{eqnarray}
f^2_h(1/C) = \alpha^2_h \cdot (1/C) + \beta^2_h, \label{eq:lquery3.1}
\end{eqnarray}
where $\alpha^2_h=W[v_L(u), v_h]$ and $\beta^2_h=d(v_h,v_R(u))\tau$.
Then the equation (\ref{eq:lquery3}) can be rewritten as
\begin{eqnarray}
\max_{v_h \in V[v_L(u),v_R(u)]} f^2_h(1/C). \label{eq:lquery3.2}
\end{eqnarray}
Here $f^2_h(1/C)$ is a linear function in $1/C$ and $\alpha^2_h$ is increasing in $h$, we thus can compute the upper envelope of $f^2_h(1/C)$ for all $v_h \in V[v_L(u),v_R(u)]$ in $O(|{\cal T}(u)|)$ time,
which is equivalent to the left capacity table of $u$ (similarly to (a)).

\begin{lemma}\label{lem:searchTree}
Given a dynamic path network with $n$ vertices and general edge capacities,
we can construct its CC-tree, ${\cal T}$, in $O(n\log n)$ time.
\end{lemma}
\begin{proof}
We construct 
two sorted lists of capacities,
two weight tables and two capacity tables 
at every node $u$ one by one (which does not need to be performed bottom up).
As mentioned above, these all can be constructed in $O(|{\cal T}(u)|)$ time.
For a non-negative integer $d$, let $U(d)$ denote a set of nodes of ${\cal T}$ such that each node $u \in U(d)$ is located at depth $d$ from root $\rho$
(see Fig.~\ref{fig:CCtree}).
Therefore, letting $h$ be the height of ${\cal T}$,
the total time required to construct ${\cal T}$ can be represented as $\sum_{d=0}^h \sum_{u \in U(d)} O(|{\cal T}(u)|)$.
We here have $h = O(\log n)$ and $\sum_{u \in U(d)} O(|{\cal T}(u)|) = O(n)$ since for a fixed $d$, ${\cal T}(u)$ for all $u \in U(d)$ are vertex-disjoint,
thus the total time is $O(n\log n)$.
%
\end{proof}


\section{Two main tasks}\label{sec:2tasks}
There are two useful tasks that we can call upon repeatedly.
{\tt Task 1} is to find a maximal subpath $P[v_a,v_d]$, given the starting vertex $v_a$,
such that we can place a 1-sink on it to enable all the evacuees to evacuate
to it within time $t$.
{\tt Task 2} is to find a 1-sink on a given subpath $P[v_i,v_j]$.
We want to construct an algorithm for each of these tasks.

In the two algorithms, 
given a subpath $P[v_i, v_j]$ and a node $u$ of $\cal T$ spanning a subpath $P[v_l,v_r]$ (i.e., $v_l=v_L(u)$ and $v_r=v_R(u)$) of $P[v_i,v_j]$,
we need to compute (\ref{eq:lquery}) with $W=W[v_i,v_{l-1}]$ and $C=c(v_r,v_j)$,
and (\ref{eq:rquery}) with $W=W[v_{r+1},v_j]$ and $C=c(v_i,v_l)$.
We first show the following lemma.
\begin{lemma}\label{lem:candidate}
Assume that the CC-tree, ${\cal T}$, is available.
Then, given a subpath $P[v_i, v_j]$ and a node $u$ of $\cal T$ spanning a subpath $P[v_l,v_r]$ of $P[v_i,v_j]$,
the L-critical candidate and the R-critical candidate of $P[v_i, v_j]$ belonging to $P[v_l, v_r]$ can be computed in $O(\log n)$ time.
\end{lemma}
\begin{proof}
We only prove the case of the L-critical candidate,
letting $W=W[v_i,v_{l-1}]$ and $C=c(v_r,v_j)$
(the proof for the R-critical candidate is symmetric).

Suppose that there exists an integer $h^*$ satisfying $l \le h^* \le r-2$ such that $c(v_{h^*},v_r) \le C$ and $c(v_{h^*+1},v_r) > C$
(if does not exist, $c(v_h,v_r) \le C$ for every $h$ satisfying $l \le h \le r-1$
or $c(v_h,v_r) > C$ for every $h$ satisfying $l \le h \le r-1$).
Note that such $h^*$ uniquely exists since $c(v_h,v_r)$ is increasing in $h$.
We first separate $P[v_l, v_r]$ to two subpaths $P_1=P[v_l, v_{h^*}]$ and $P_2=P[v_{h^*+1}, v_r]$,
which can be done in $O(\log n)$ time by binary search over the sorted list of capacities stored at $u$.
Letting $V_1=V[v_l, v_{h^*}]$ and $V_2=V[v_{h^*+1}, v_r]$, we then consider
\begin{eqnarray}
cost^L_u(W,C,P_1)	&=&\max_{v_h \in V_1} \left\{d(v_h,v_r)\tau + \frac{W +W[v_l, v_h]}{\min\{c(v_h,v_r), C\}}\right\} \nonumber \\
				&=&\max_{v_h \in V_1} \left\{d(v_h,v_r)\tau + \frac{W +W[v_l, v_h]}{c(v_h,v_r)}\right\}, \label{eq:lqueryp1}
\end{eqnarray}
and 
\begin{eqnarray}
cost^L_u(W,C,P_2)	&=&\max_{v_h \in V_2} \left\{d(v_h,v_r)\tau + \frac{W +W[v_l, v_h]}{\min\{c(v_h,v_r), C\}}\right\} \nonumber \\
				&=&\max_{v_h \in V_2} \left\{d(v_h,v_r)\tau + \frac{W +W[v_l, v_h]}{C}\right\} \nonumber \\
				&=&\max_{v_h \in V_2} \left\{d(v_h,v_r)\tau + \frac{W[v_l, v_h]}{C}\right\} + \frac{W}{C}. \label{eq:lqueryp2}
\end{eqnarray}
Note that $cost^L_u(W,C)=\max \{ cost^L_u(W,C,P_1),cost^L_u(W,C,P_2) \}$.
By binary search over the left weight table of $u$, we can identify a vertex $v^*_1$ maximizing $\{d(v_h,v_r)\tau + (W +W[v_l, v_h])/c(v_h,v_r)\}$ for $v_h \in V_1 \cup V_2$ in $O(\log n)$ time.
Similarly, using the left capacity table of $u$,
we can identify a vertex $v^*_2$ maximizing $\{d(v_h,v_r)\tau + W[v_l, v_h]/C\}$ for $v_h \in V_1 \cup V_2$ in $O(\log n)$ time.
Note that if $v^*_1 \in V_2$, 
\begin{eqnarray}
cost^L_u(W,C,P_1)	&\le& d(v^*_1,v_r)\tau + \frac{W +W[v_l, v_h]}{c(v^*_1,v_r)} \nonumber \\
				&<& d(v^*_1,v_r)\tau + \frac{W +W[v_l, v_h]}{C} \le cost^L_u(W,C,P_2). \nonumber
\end{eqnarray}
and if $v^*_2 \in V_1$, 
\begin{eqnarray}
cost^L_u(W,C,P_2)	&\le& d(v^*_2,v_r)\tau + \frac{W +W[v_l, v_h]}{C} \nonumber \\
				&\le& d(v^*_1,v_r)\tau + \frac{W +W[v_l, v_h]}{c(v^*_2,v_r)} \le cost^L_u(W,C,P_1), \nonumber
\end{eqnarray}
which implies that $v^*_1 \in V_2$ and $v^*_2 \in V_1$ never occur simultaneously.
Therefore, if $v^*_1 \in V_2$, $v^*_2 \in V_2$ and $cost^L_u(W,C,P_1)< cost^L_u(W,C,P_2)$, thus $v^*_2$ is the L-critical candidate of $P[v_i, v_j]$.
If $v^*_2 \in V_1$, $cost^L_u(W,C,P_1) \ge cost^L_u(W,C,P_2)$ and $v^*_1 \in V_1$ holds, thus $v^*_1$ is the L-critical candidate of $P[v_i, v_j]$.
Otherwise $v^*_1 \in V_1$ and $v^*_2 \in V_2$, then $v^*_1$ achieves $cost^L_u(W,C,P_1)$ and $v^*_2$ also achieves $cost^L_u(W,C,P_2)$, respectively.
We then compare these two costs and choose one whose cost is larger.
\end{proof}

\subsection{Basic algorithms}\label{sec:2algorithms}
Let us first design an algorithm for {\tt Task 1},
referring to Fig.~\ref{fig:upDown},
which shows a part of the CC-tree $\cal T$.
\begin{figure}[ht]
\centering
\includegraphics[height=3cm]{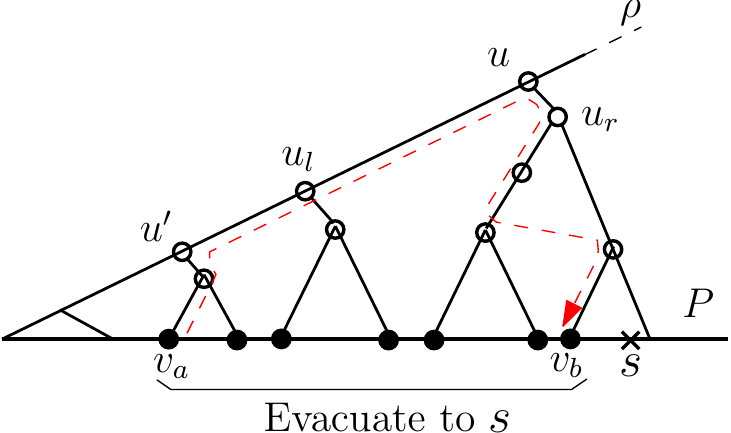}
\hspace{5mm}
\includegraphics[height=3cm]{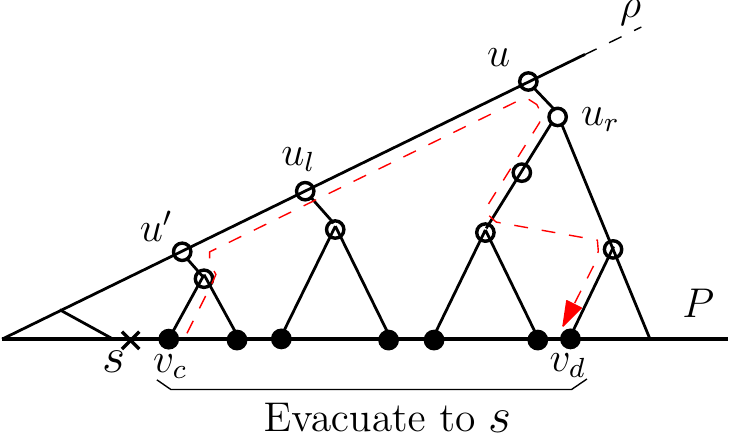}
\caption{
{\em Left}: Looking for rightmost $v_b$ such that the evacuees from $P[v_a,v_b]$
can evacuate to sink $s$ within time $t$;
{\em Right}: Looking for rightmost $v_d$ whose evacuees can evacuate to sink $s$ within time $t$.
}
\label{fig:upDown}
\end{figure}
Here is an informal description of the algorithm.

\begin{algorithm} {\tt Isolate-subpath$(t,v_a)$}
~
\begin{enumerate}
\item
Start at leaf $u=v_a$ of ${\cal T}$,\footnote{See the left figure in Fig.~\ref{fig:upDown}.}
 and move up towards its root $\rho$.
At each node $u$ visited,
do $L$-$test(P[v_a,v_R(u)],t)$, i.e., 
compute the cost of the L-critical vertex of subpath $P[v_a,v_R(u)]$ w.r.t. $v_R(u)^+$, say $\Theta_L(v_a,v_R(u))$.
If $L$-$test(P[v_a,v_R(u)],t)$ returns "yes",
i.e., $\Theta_L(v_a,v_R(u)) \le t$, 
then move to its parent node $p$ and set $u=p$.
If $L$-$test(P[v_a,v_R(u)],t)$ returns "no",
then move to the right child node $u_r$ and
set $u=u_r$.
Start moving down towards a leaf.
\item
At each node visited during moving down,
do $L$-$test(P[v_a,v_R(u)],t)$.
If "yes",
then move to its parent's right child node $p_r$ and set $u=p_r$.
If "no",
then move to the left child node $u_l$ and
set $u=u_l$.
If $u$ comes to a leaf, say $v_b$, locate the 1-sink $s \in (v_b,v_{b+1}]$ to the left as much as possible.
%
%
\item
Start from the vertex $v_c$ that lies immediately to the right of $s$.\footnote{See the right figure in Fig.~\ref{fig:upDown}.}
Performing an up-down search similar to above 1 and 2 so that $R$-$test(P[v_c,v_R(u)],s,t)$ is done at each visited node $u$,
determine the rightmost vertex $v_d$ whose evacuees can reach sink $s$ within time $t$.
\end{enumerate}
\end{algorithm}

\begin{lemma}\label{lem:1subpathA}
Assume that the CC-tree, ${\cal T}$, is available.
Then {\tt Isolate-subpath$(t,v_a)$} runs in $O(\log^3 n)$ time.
\end{lemma}
\begin{proof}
At each node $u$, {\tt Isolate-subpath$(t,v_a)$} carries out $L$-$test(P[v_c,v_R(u)],t)$ or $R$-$test(P[v_c,v_R(u)],s,t)$.
Each of them needs to find the critical vertex by comparing $O(\log n)$ critical candidates.
A critical candidate can be computed in $O(\log n)$ time by Lemma~\ref{lem:candidate}.
Since {\tt Isolate-subpath$(t,v_a)$} visits $O(\log n)$ nodes,
the total time is $O(\log^3 n)$.
\end{proof}

\medskip\noindent
\begin{algorithm}{\tt Find-1sink$(v_i,v_j)$}
~
\begin{enumerate}
\item
Let $u$ be the node where the two paths $\pi(v_i,\rho)$ and $\pi(v_j,\rho)$ meet.
\item
If the L-critical vertex of $P[v_i,v_R(u_l)]$ and the R-critical vertex of $P[v_L(u_r),v_j]$
have the same cost\footnote{These costs can be computed in $O(\log^3 n)$ time as we saw above.}
at some point $x$ on the edge $(v_R(u_l),v_L(u_r))$,
then return $x$ as the 1-sink.
\item
If the L-critical vertex has a higher (resp. lower) cost than the R-critical vertex 
at every point on edge $(v_R(u_l),v_L(u_r))$,
then let $u=u_l$ (resp. $u=u_r$) and repeat Step 2,
using the new $u_l$ and $u_r$.
\end{enumerate}
\end{algorithm}

Using the arguments similar to those in the proof of Lemma~\ref{lem:1subpathA},
we can prove the following lemma.

\begin{lemma}\label{lem:1subpathB}
Assume that the CC-tree, ${\cal T}$, is available.
Then {\tt Find-1sink$(v_i,v_j)$} finds the 1-sink on a given subpath $P[v_i,v_j]$ in $O(\log^3 n)$ time.
\end{lemma}

\subsection{$(t,k)$-feasibility test}\label{sec:feasibility}
Our approach is to find the maximal subpath from the left end of $P$
for which a 1-sink can achieve completion time $t$.

\begin{lemma}\label{lem:feasibility}
Given a dynamic path network with $n$ vertices,
assume that its CC-tree, ${\cal T}$, is available.
Then we can test its $(t,k)$-feasibility in $O(k\log^3 n)$ time.
\end{lemma}
\begin{proof}
Starting at the leftmost vertex $v_1$ of $P$,
invoke {\tt Isolate-subpath$(t)$},
which isolates the first subpath in $O(\log^3 n)$ time,
and remove it from $P$.
We repeat this at most $k-1$ more times on the remaining subpath, spending $O(k\log^3 n)$ time.
The problem instance is $(t,k)$-feasible if and only if the rightmost vertex $v_n$ belongs
to the last isolated subpath.
\end{proof}

\subsection{Uniform edge capacity case}\label{sec:uniform}
The problem is much simplified if the edges have the same capacity.
In particular, we can compute the critical vertex of a subpath resulting from
concatenating two subpaths in constant time.
At each node $u$ of ${\cal T}$ bottom up,
we compute and record the L- and R-critical vertices of
$P[v_L(u),v_R(u)]$ w.r.t. $v_R(u)^+$ and their costs,
based on the following lemma.
\begin{lemma}\label{lem:merge1}
{\rm \cite{higashikawa2015a}} For a node $u$ of CC-tree $\cal T$,
let $v_L(u_l)=v_h$, $v_R(u_l)=v_j$, $v_L(u_r)=v_{j+1}$, and $v_R(u_r)=v_l$,
and assume that the critical vertices,
$c_L^{[h,j]}$, $c_R^{[h,j]}$, $c_L^{[j+1,l]}$, and $c_R^{[j+1,l]}$
have already been computed.
\begin{enumerate}
\item[(a)]
The L-critical vertex $c_L^{[h,l]}$ is either $c_L^{[h,j]}$ or $c_L^{[j+1,l]}$.
\item[(b)]
The R-critical vertex  $c_R^{[h,l]}$ is either $c_R^{[h,j]}$ or $c_R^{[j+1,l]}$.
\end{enumerate}
\end{lemma}

For example,
to the cost of the L-critical cluster $C_L^{[h,j]}$ of $P[v_h,v_j]$,
we add the distance cost $d(v_j,v_l)\tau$,
and to the cost of the L-critical cluster $C_L^{[j+1,l]}$ of $P[v_{j+1},v_l]$ we just add
$W[v_h,v_j]/c$ to compute its new cost.
The L-critical cluster of the combined path $P[v_h,v_l]$ is whichever is larger.

\begin{lemma}\label{lem:uppEnvTree}
Given a dynamic path network with $n$ vertices and uniform edge capacities,
we can construct search tree ${\cal T}$ in $O(n)$ time and $O(n)$ space.
\end{lemma}

Thanks to Lemma~\ref{lem:uppEnvTree},
Algorithm {\tt Isolate-subpath$(t,v_a)$} runs in $O(\log n)$ time.
We thus have
\begin{lemma}\label{lem:feasibility-uniform}
Given a dynamic path network with $n$ vertices and uniform edge capacities,
assume that its  search tree, ${\cal T}$, is available.
Then we can test its $(t,k)$-feasibility in $O(k\log n)$ time.
\end{lemma}

Algorithm~{\tt Find-1sink$(v_i,v_j)$} also runs in $O(\log n)$ time,
which implies
\begin{lemma}\label{lem:1-sink}
Given a dynamic path network with $n$ vertices and uniform edge capacities,
assume that its  search tree, ${\cal T}$, is available.
Then we can find the $1$-sink on subpath $P[v_i,v_j]$ in $O(\log n)$ time.
\end{lemma}


\section{Optimization}\label{sec:optimization}
\begin{lemma}{\rm \cite{arumugam2016}}\label{lem:golin}
If $(t,k)$-feasibility can be tested in $T(t,k)$ time,
then the $k$-sink can be found in $O(T(t,k) \cdot k\log n)$ time,
excluding the preprocessing time.
\end{lemma}
By Lemma~\ref{lem:searchTree}
it takes $O(n\log n)$ time to construct $\cal T$ with weight and capacity data,
and $T(t,k)=O(k\log^3 n)$ by Lemma~\ref{lem:feasibility}.
We thus have
\begin{theorem}\label{thm:optimization-uniformA}
Given a dynamic path network with $n$ vertices,
we can find an optimal $k$-sink in $O(n\log n+k^2\log^4 n)$ time.
\end{theorem}

Based on Lemmas~\ref{lem:uppEnvTree} and \ref{lem:feasibility-uniform},
Megiddo's theorem in~\cite{megiddo1979} implies 
(it also follows Lemma~\ref{lem:golin})
\begin{theorem}\label{thm:optimization-uniformB}
Given a dynamic path network with $n$ vertices and uniform edge capacities,
we can find an optimal $k$-sink in $O(n + k^2\log^2 n)$ time.
\end{theorem}

\subsection{Sorted matrix approach}\label{sec:sorted matrix}
Let $OPT(l, r)$ denote the evacuation time for the optimal 1-sink on subpath $P[v_l, v_r]$.
Define an $n\times n$ matrix $A$ whose entry $(i, j)$ entry is given by
\begin{equation}
A[i,j] = \left\{ \begin{array}{ll}\label{eqn:Aij}
                     {\it OPT}(n-i+1, j) 	&\mbox{~~if~}   n-i+1\leq  j\\
                     0 				&\mbox{~~otherwise.} 
                      \end{array}
                      \right.                    
\end{equation}

It is clear that
matrix $A$ includes ${\it OPT}(l, r)$  for every pair of integers $l$ and $r$ such that $1 \le l \le r \le n$.
There exists a pair of integers  $l$ and $r$ such that ${\it OPT}(l, r)$ is the evacuation time
for the optimal $k$-sink on the whole path. Then $k$-sink location problem can be written as: 
``Find the smallest $A[i,j]$ such that the given problem instance is $A[i,j]$-feasible.''

A matrix is called a {\em sorted matrix} if each row and column of it is sorted
in the nondecreasing order.
In \cite{frederickson1991a,frederickson1983},
Frederickson et al. show how to search for such a minimum in a sorted matrix.
The following lemma is implicit in their papers.

\begin{lemma}\label{lem:frederickson}
Suppose that $A[i,j]$ can be computed in $f(n)$ time, and feasibility can be tested in $g(n)$ time.
Then we can solve the $k$-sink problem in $O(n f(n) + g(n) \log n)$ time. 
\end{lemma}

We have $f(n)=O(\log^3 n)$ by Lemma~\ref{lem:1subpathB},
and $g(n)$ can be $O(n \log^2 n)$ by scanning path $P$ from left to right.
Lemma~\ref{lem:frederickson} thus implies
\begin{theorem}\label{thm:optimization-general}
Given a dynamic path network with $n$ vertices and general edge capacities,
we can find an optimal $k$-sink in $O(n \log^3 n)$ time.
\end{theorem}

In the uniform capacity case,
we can show that $f(n)=O(\log n)$,
and $g(n)$ can be $O(n)$ by scanning the path from left to right.
Lemma~\ref{lem:frederickson} thus implies
\begin{theorem}\label{thm:optimization-uniform2}
Given a dynamic path network with $n$ vertices and uniform edge capacities,
we can find the $k$-sink in $O(n \log n)$ time.
\end{theorem}


\section{Conclusion and discussion}\label{sec:conclusion}
We have shown that on dynamic path networks with $n$ vertices,
the $k$-sink can be found in $O(\min\{n + k^2\log^4n, n\log^3n\})$ time,
which is sub-quadratic.
If the edges have the same capacity,
we can solve the $k$-sink problem in $O(\min\{n + k^2\log^2n,n\log n\})$ time.
These results improve upon the previously best algorithms~\cite{arumugam2016,higashikawa2015a}
for all values of $k$.

\bibliographystyle{splncs03}
\bibliography{stacs2017}
\end{document}